\documentclass[aps,pra,reprint,groupedaddress]{revtex4-1} %reprint

\usepackage{amsmath,amsthm,amssymb}
\usepackage{graphicx}
\usepackage{dcolumn}
\usepackage{braket}
\usepackage{appendix}
\usepackage[utf8]{inputenc}
\usepackage{xcolor}
\usepackage{mathrsfs}

\theoremstyle{plain}
\newtheorem{theorem}{Theorem}

\newtheorem{corollary}[theorem]{Corollary}

\theoremstyle{remark}

\theoremstyle{definition}

\begin{document}

\title{Guaranteed convergence for a class of coupled-cluster methods based on Arponen's extended theory}

\author{Simen Kvaal}
\email{simen.kvaal@kjemi.uio.no}
\author{Andre Laestadius}
\author{Tilmann Bodenstein}
\affiliation{Hylleraas Centre for Quantum Molecular Sciences, Department of Chemistry, University of Oslo, P.O. Box 1033 Blindern, N-0315 Oslo, Norway }
\date{Friday March 13, 2020}

\begin{abstract} 
A wide class of coupled-cluster methods is introduced, based on Arponen's extended coupled-cluster theory. This class of methods is formulated in terms of a coordinate transformation of the cluster operators. The mathematical framework for the error analysis of coupled-cluster methods based on Arponen's bivariational principle is presented, in which the concept of local strong monotonicity of the flipped gradient of the energy is central. A general mathematical result is presented, describing sufficient conditions for coordinate transformations to preserve the local strong monotonicity. The result is applied to the presented class of methods, which include the standard and quadratic coupled-cluster methods, and also Arponen's canonical version of extended coupled-cluster theory. Some numerical experiments are presented, and the use of canonical coordinates for diagnostics is discussed.
\end{abstract}

\pacs{To be added} 
\keywords{Coupled-cluster method, extended coupled-cluster method, error analysis, electronic-structure theory}

\maketitle

\section{Introduction}

It is with delight that the authors dedicate this work to Professor J\"urgen Gau\ss{} on the occasion of his
sixtieth birthday. In the spirit of his pursuit of scientific rigor, especially the attention to detail in coupled-cluster (CC) theory, we here present a mathematical study of some alternative formulations based on Arponen's extended CC (ECC) method~\cite{Arponen1983,Arponen1987}. 

Our perspective and our assumptions are natural for the
electronic structure problem of molecules
in the Born--Oppenheimer approximation, but should also be useful in a
more general setting. The collection of methods is defined by substitution of the dual exponential $e^{\Lambda^\dag}$ by a Taylor polynomial of fixed degree $n$ in the exact ECC energy functional in canonical (C) and non-canonical (NC) coordinates introduced by Arponen (see Eqs.~\eqref{eq:NC-ECC} and \eqref{eq:C-ECC}).  This can be viewed as a coordinate transformation, and leads to two hierarchies of models, the NC-ECC$(n)$ class using non-canonical coordinates, and the C-ECC$(n)$ class using canonical coordinates (see Eq.~\eqref{eq:hierarchy}).  Our mathematical results imply that when the cluster operators in the energy functionals are not truncated, all these models are exact and equivalent to the Schr\"odinger equation. Moreover, Galerkin approximations (i.e., generic truncation schemes that can approach the untruncated limit) will converge under certain relatively mild single-reference type conditions.

The various forms of CC methods are today among the
most widely used for wavefunction-based calculations on manybody
systems. The main idea stems from Hubbard's exponential
parameterization of the wavefunction based on cluster
operators in manybody perturbation theory~\cite{Hubbard1957}, which was
taken as starting point for \emph{ab initio} treatments by Coester and
K\"ummel for nuclear structure calculations in the
1950s~\cite{Coester1958,Coester1960}. The modern form of standard
CC theory was developed by, among others, Sinano\u{g}lu,
Paldus and Cizek in the 1960s~\cite{Paldus2005} and the CC
method with singles, doubles and perturbative triples [CCSD(T)]
today constitutes ``the gold standard of quantum chemistry'' due to
its excellent balance between computational cost and
accuracy~\cite{Bartlett2007}. In nuclear structure calculations the
same method has gained traction in the last decade, providing
excellent predictive power for light to medium
nuclei~\cite{Hagen2014}. Coupled-cluster theory has also been applied to
superconductivity~\cite{Emrich1984}, lattice gauge
theory~\cite{McKellar2000}, and systems of trapped bosons such as
Bose--Einstein condensates~\cite{Cederbaum2006}. These examples and
the cited works are by no means exhaustive, but serve to illustrate
the flexibility of the CC formalism.

In the early 1980s, Arponen introduced a novel concept into CC theory, namely
\emph{the bivariational principle}~\cite{Arponen1982,Arponen1983}, resulting in
the ECC method~\cite{Arponen1983,Arponen1987,Arponen1987b}, and an interpretation of standard CC theory and ECC theory as
variational methods in a more general sense, i.e., they are bivariational. Today, the view of the CC energy functional as a Lagrangian is standard in quantum chemistry~\cite{Helgaker1988}. 

However, the ECC method has seen little use in chemistry due to its
immense complexity, even for truncated versions. In physics, on the
other hand, the ECC model has advantages over standard CC theory that
can make it very useful. To illustrate, the ECC method correctly
describes symmetry breaking in the Lipkin--Meshkov--Glick quasispin
model of collective monopole vibrations in
nuclei~\cite{Arponen1982,Bishop1991}, in contrast to the standard CC
method, which cannot. For the electronic-structure problem in quantum
chemistry, the standard CC model fails dramatically to reproduce
dissociation curves of even simple dimers like $\mathrm{N}_2$, while the ECC method performs quite well~\cite{Cooper2010,Evangelista2011}, as do the quadratic CC doubles model introduced by Van Voorhis and Head--Gordon~\cite{VanVorhiis2000a,Byrd2002}.
The latter approach has asymptotic cost similar to standard CC with singles and doubles (CCSD). 
Thus, we conclude that the ECC method is still worthwhile to study, and approximate forms such as studied in this article, may still prove to be useful.

The non-canonical and canonical hierarchies (N)C-ECC($n$) introduced in this article turn out to be equivalent, and give identical predictions, when truncated with an excitation-rank complete scheme. On the other hand, the working equations are different and in fact cheaper in the canonical case, albeit marginally. An example is the NC-ECC(1)SD method, i.e., the standard CCSD approach, and the C-ECC(1)SD method, which are equivalent. We also raise the question about diagnostics for practical calculations, and show some numerical evidence that diagnostics can favorably be done using canonical coordinates, even if the computations are done in the usual manner using noncanonical variables.

Another well-known special case is NC-ECC(2)D, the quadratic coupled-cluster (QCC) method~\cite{VanVorhiis2000a,Byrd2002}, which is also equivalent in the canonical and non-canonical versions. Furthermore, the  perfect-pairing (PP) hierarchy~\cite{Lehtola2016} of amplitude truncation schemes can be applied to our methods. The PP hierarchy are approximations to the complete-active space self-consistent field (CASSCF) method, including only a tiny subset of even-rank amplitudes combined with orbital-optimization, the latter which we disregard here. The corresponding canonical and non-canonical formulations (N)C-ECC($1$)PPH are inequivalent. The $n>1$ versions could also be interesting in their own right, as investigated by Byrd and coworkers in the case of QCC~\cite{Byrd2002}.

The remainder of the article is organized as follows: In Section~\ref{sec:ncECC}, we introduce the bivariational principle and the mathematical setting of local analysis of CC methods.
The key concept of our analysis is the notion of local strong monotonicity of the flipped gradient of a smooth bivariational energy functional (see Eq.~\eqref{eq:flipped-gradient}).
The usefulness of this property is presented in Theorem~\ref{thm:z}, where local uniqueness and quadratic error estimates are established in a very general setting using Zarantonello's Theorem from nonlinear monotone operator theory~\cite{Zarantonello1960,Zeidler1990}. Next, Theorem~\ref{thm:NC-ECCmono} summarizes the main results of Ref.~\onlinecite{Laestadius2018}, where strong montonicity is proven for the non-canonical ECC method. For a recent review on monotonicity in CC theory we refer to \cite{Laestadius2019}, where this property is linked to spectral gaps of the systems under study. Section \ref{sec:coordinate-transformations} presents 
the idea of monotonicity-preserving coordinate transformations. Our main result, Theorem~\ref{thm:map}, is a change-of-coordinates result. When combined with Theorems~\ref{thm:z} an \ref{thm:NC-ECCmono}, the analysis of (N)C-ECC($n$) follows in Corollary~\ref{cor:methods}. Our tools rely heavily on the functional analytic formulation of cluster operators and the Schr\"odinger equation developed by Rohwedder and Schneider~\cite{Schneider2009,Rohwedder2013,Rohwedder2013b}. 
In Section~\ref{sec:numerics}, we perform some numerical experiments to elucidate some aspects of the (N)C-ECC($n$) hierarchies, before we finish with some concluding remarks in Section~\ref{sec:conclusion}.

\section{The non-canonical extended coupled-cluster model}
\label{sec:ncECC}
\subsection{Bivariational principle}

The starting point is a generalization of
the Rayleigh--Ritz variational principle to operators that are not
\emph{necessarily} self-adjoint (Hermitian in the finite-dimensional
case). For simplicity, we assume a real Hilbert space $\mathcal{H}$. Given a system Hamiltonian $\hat{H} : D(\hat{H}) \to \mathcal{H}$, where $D(\hat{H})\subset\mathcal{H}$ is dense, we define a bivariate Rayleigh quotient, 
$\mathcal{E}_\mathrm{bivar}: \mathcal H \oplus \mathcal H \to \mathbb R$, 
\begin{equation}
  \mathcal{E}_\text{bivar}(\psi,\tilde\psi) =
  \frac{\braket{\tilde\psi,\hat{H}\psi}}{\braket{\tilde\psi,\psi}},
  \quad \braket{\tilde\psi,\psi}\neq 0. \label{eq:quotient}
\end{equation}
Requiring the functional $\mathcal{E}_\text{bivar}$ to be stationary at $(\psi_*,\tilde{\psi}_*)$ with respect
to arbitrary variations in the two wavefunctions leads to the
conditions $\braket{\tilde{\psi}_*,\psi_*}\neq 0$, $\hat{H}\psi_*=E_*\psi_*$ and
$\hat{H}^\dag\tilde{\psi}_*=E_*\tilde{\psi}_*$, with $E_* = \mathcal{E}_\text{bivar}(\psi_*,\tilde{\psi}_*)$, i.e., the right and left eigenvalue problem for $\hat{H}$. If $\hat{H}$ is self-adjoint, 
the eigenfunctions are identical up to normalization.
The introduction of two independent wavefunctions therefore might seem to complicate matters. However, the bivariate Rayleigh quotient
$\mathcal{E}$ allows \emph{distinct} approximations of $\psi$ and
$\tilde{\psi}$, introducing more flexibility for approximate
schemes. Moreover, the \emph{state} defined is a (non-Hermitian) density operator, which is unique,
\begin{equation*}
    \rho = \frac{\ket{\psi}\bra{\tilde{\psi}}}{\braket{\tilde{\psi}|\psi}}.
\end{equation*}
When determined variationally, the Hellmann--Feynman theorem~\cite{Feynman1939} gives well-defined physical predictions in terms of $\rho$.

As is common in analysis of partial differential equations~\cite{ReedSimonI,Schmuedgen2012}, we pass to a weak formulation, which in this case is \emph{equivalent} to the strong formulation outlined above. Under the assumption that $\hat{H}$ is below bounded, we can introduce a unique extension $H : \mathcal{X} \to \mathcal{X}'$ (dual space), where $\mathcal{X} \subset \mathcal{H}$ is a dense
subspace, a Hilbert space with norm $\Vert \cdot \Vert_\mathcal{X}$, continuously embedded in $\mathcal{H}$. It follows that $\mathcal{H}$ is continuously
embedded in $\mathcal{X}'$, and we have a scale
of spaces with dense embeddings, $\mathcal{X}\hookrightarrow
\mathcal{H}\hookrightarrow \mathcal{X}'$. The operator
$H$ is bounded (i.e., continuous), and satisfies a Gårding estimate, i.e., for some $\alpha\geq 0$ and some $\mu \in \mathbb{R}$,
\[ \braket{\psi,H\psi} \geq \alpha \|\psi\|_{\mathcal{X}}^2 + \mu \|\psi\|^2 \]
for all $\psi \in \mathcal{X}$. For the electronic-structure problem $\mathcal{X}$ can be taken to be the space of functions with finite kinetic energy.

If $\mathcal{H}$ is finite-dimensional, we can set
$\mathcal{X}\equiv \mathcal{H}$, simplifying matters a lot, and the reader may
if she or he wishes stick to this picture for simplicity, where
all operators are basically matrices. In the infinite
dimensional case, however, $\hat{H}$ is typically unbounded as an operator over
$\mathcal{H}$, and the above construction is necessary.

Under the above stated conditions, $\mathcal{E}_\text{bivar} :
\mathcal{X}\oplus\mathcal{X}\to\mathbb{R}$ is a (Fréchet) smooth map away from the singularity $\braket{\tilde{\psi},\psi}=0$, and differentiation and Taylor series exist and converge locally, allowing a certain degree of intuition to be borrowed from the finite-dimensional case. The right and left Schrödinger equations are then $\partial_{\tilde{\psi}}\mathscr{E}_\text{bivar}(\psi_*,\tilde{\psi}_*) = 0$ and $\partial_{\psi}\mathscr{E}_\text{bivar}(\psi_*,\tilde{\psi}_*) = 0$, respectively. This is the bivariational principle.

\subsection{Exponential ansatz and the ECC method}

The standard CC method is formulated relative to a fixed
reference $\phi_0\in \mathcal{X}$ on determinantal form, and by introducing a cluster
operator ${T} = {T}_1 + {T}_2 + \cdots$ with ${T}_k$
containing all excitations of rank $k$, i.e., of $k$ fermions relative to $\phi_0$, we
have the exact parameterization
\[ \psi = e^{T}\phi_0 , \]
assuming intermediate normalization, $\braket{\phi_0,\psi} =
1$.

Since all excitations commute, the cluster operators form a commutative Banach algebra under suitable conditions which we now describe~\cite{Schneider2009,Rohwedder2013}. We expand the cluster operators using amplitudes and basis operators, i.e., $T = \sum_{\mu\in \mathcal{I}} \tau_\mu X_\mu$, where $X_\mu$ excites a number $n = n(\mu)$ of fermions in the reference into the virtual space, i.e., 
\[ X_\mu = c_{a_1}^\dag c_{i_1} \cdots c_{a_n}^\dag c_{i_n}, \]
where the $i_k$ are among the occupied orbitals of $\phi_0$, and $a_k$ among the unoccupied orbitals. The set $\mathcal{I}$ is the generic set of amplitude indices. We introduce a Hilbert space $\mathcal{V}$ with norm $\|T\| = \|T\phi_0\|_{\mathcal{X}}$, which becomes a useful space for formulating abstract CC theory. Fundamental results include that any $T\in \mathcal{V}$ is a bounded operator on $\mathcal{X}$, such that, e.g., $\exp(T)$ also is a bounded operator. Moreover, $T^\dag$ is also a bounded operator, which means that we can make sense of, e.g., $\exp(-T)H\exp(T)$, and that we can represent any intermediately normalized $\psi\in \mathcal{X}$ as $\psi = e^T\phi_0$ with $T\in \mathcal{V}$ unique. Finally, all the elements of the algebra are nilpotent. The Banach algebra structure on $\mathcal{V}$ allows CC theory to be rigorously formulated in the full, infinite-dimensional case. This was the approach taken  in Ref.~\onlinecite{Laestadius2018} for a first analysis of NC-ECC theory.

Again, the finite-dimensional case may by kept in mind: In this case, cluster amplitudes are simply finite-dimensional vectors, and the existence of the exponential parameterization is a trivial result. There is no need to introduce the norm $\|T\|$, instead the Euclidean norm on the amplitudes may be used.

Any $\tilde{\psi}$ normalized according to
$\braket{\tilde{\psi},\psi} = 1$ can be represented by introducing a
second cluster operator ${\Lambda} = {\Lambda}_1 +
{\Lambda}_2 + \cdots$, viz.,
\[ \tilde\psi = e^{-{T}^\dag} e^{\Lambda}\phi_0 . \]
Plugging into the bivariate Rayleigh quotient, we obtain the energy functional $\mathcal{E}_\mathrm{NC-ECC} : \mathcal{V}\oplus\mathcal{V}\to\mathbb{R}$ of the non-canonical ECC method, given by
\begin{equation}
  \mathcal{E}_{\text{NC-ECC}}(T,\Lambda) =
  \braket{\phi_0,e^{\Lambda^\dag} e^{-{T}} {H} e^{T}\phi_0}. \label{eq:NC-ECC}
\end{equation}
This map is everywhere smooth, and its critical points $(T_*,\Lambda_*)$ are equivalent to the Schrödinger equation and its dual: Under the assumption that the eigenfunctions can be normalized according to $\braket{\phi_0,\psi_*}=\braket{\tilde{\psi}_*,\psi_*}=1$, $\psi_*$ and $\tilde{\psi}_*$ solve the Schr\"odinger equation and its dual if and only if
\begin{equation}
    \frac{\partial{\mathcal{E}_\text{NC-ECC}(T_*,\Lambda_*)}}{\partial \Lambda} = 0 \quad\text{and}\quad \frac{\partial{\mathcal{E}_\text{NC-ECC}(T_*,\Lambda_*)}}{\partial T} = 0. \label{eq:ampeq}
\end{equation}
Assuming that the eigenvalue $E_* = \mathcal{E}(T_*,\Lambda_*)$ is nondegenerate, $(T_*,\Lambda_*)$ is easily seen to be locally unique.

\subsection{Truncations and monotonicity analysis}

The NC-ECC energy is just one out of many possible parameterizations of the exact bivariate Rayleigh quotient $\mathcal{E}_\text{bivar}$. In this section, we take a more abstract approach and consider a general energy functional $\mathcal{E} : \mathcal{V}\oplus\mathcal{V} \to \mathbb{R}$, obtained by some exact parameterization of $(\psi,\tilde{\psi})$ by means of the space $\mathcal{V}\oplus\mathcal{V}$, i.e., by a pair of cluster operators $(T,\Lambda)$. We will discuss several such functionals in Sec.~\ref{sec:coordinate-transformations}, obtained from the NC-ECC functional by coordinate transformations.

Only in rare cases can the amplitude equations~\eqref{eq:ampeq} be solved exactly. Introduce therefore a discretized space $\mathcal{V}_d \subset \mathcal{V}$ of finite dimension by truncating the amplitude index set $\mathcal{I}_d\subset \mathcal{I}$, that is, $T_d\in \mathcal{V}_d$ if and only if
\begin{equation}
    T_d = \sum_{\mu\in\mathcal{I}_d} \tau_{d,\mu} X_\mu \in \mathcal{V}_d. 
\end{equation} 
The set $\mathcal{I}_d$ is typically defined by the restriction of the excitations to a finite virtual space (a finite basis), and to a finite excitation rank.  In the chemistry literature, the excitation hierarchy for a given basis is traditionally denoted singles (S), doubles (D), and so on. In the ECC literature, one typically speaks of the SUB$n$ approximation, with $n$ being the maximum rank. 

When the discrete space is established, we define a discrete solution by the stationary conditions of the restricted energy function $\mathcal{E}_d = \mathcal{E}\restriction_{\mathcal{V}_d\oplus\mathcal{V}_d}$. 
The stationary equations take the form
\begin{equation}
  \frac{\partial{\mathcal{E}(T_{d*},\Lambda_{d*})}}{\partial \lambda_\mu} = \frac{\partial{\mathcal{E}(T_{d*},\Lambda_{d*})}}{\partial \tau_\mu} = 0, \label{eq:ampeq-h}
\end{equation}
for all $\mu \in \mathcal{I}_d$. 

It is not \emph{necessary} to use the traditional truncation scheme outlined here; any increasing sequence of subspaces $\mathcal{V}_d\subset\mathcal{V}$, with $d$ a parameter, that can approximate elements in $\mathcal{V}$ arbitrarily well by increasing $d$ can be used. We let $\mathrm{dist}(v,\mathcal{V}_d)$ be the distance from $v$ to $\mathcal{V}_d$ measured with respect to the norm of $\mathcal{V}$. Consequently, for all $v\in \mathcal V$ we have $\mathrm{dist}(v,\mathcal V_d) \to 0$ as $d\to+\infty$.
Such a sequence of spaces is referred to as a Galerkin sequence. Other options than the traditional truncation schemes are explicitly correlated methods \cite{Hattig2012} and complete-active space methods \cite{Adamowicz2000,Kowalski2018,Lehtola2016} such as the PP hierarchy.

An often overlooked point in the physics literature is the fact that convergence of the \emph{equations} does not in general imply convergence of their \emph{solutions}. An important question is therefore whether the discrete critical points $(T_{d*},\Lambda_{d*})$ converge to the exact critical points $(T_*,\Lambda_*)$ as $d\to +\infty$. This would imply that the energy converges too, and in a quadratic manner due to the critical point formulation.

\emph{Monotonicity} is an important notion in connection to the local analysis of the CC method and its variations~\cite{Schneider2009,Rohwedder2013,Rohwedder2013b,Laestadius2018,faulstich2018,Laestadius2019}. The use of montonicity in the analysis of the standard CC method was introduced by Schneider and Rohwedder~\cite{Schneider2009,Rohwedder2013b}. It allows the establishment of locally unique solutions of the Galerkin problem and is therefore important for the motivation of numerical implementations. As such it is a fundamental result of the CC method's practical usage in quantum chemistry. It also connects spectral gaps, e.g., HOMO-LUMO gap, to stability constants within the analysis~\cite{Laestadius2019}. (See also the steerable CAS-ext gap connected to the tailored CC method~\cite{Kinoshita2005} that treats quasi-degenerate systems~\cite{faulstich2019numerical}.)

The particular monotonicity property that is key for this presentation is as follows: A function $F: \mathcal{V}\oplus\mathcal{V} \to \mathbb {\mathcal{V}'\oplus\mathcal{V}'}$, $Z \mapsto F(Z)$ is locally strongly monotone at $Z_*$ if there is an open ball $U \subset \mathcal{V}\oplus\mathcal{V}$ containing $Z_*$, such that for all $Z_1,Z_2\in U$, we have
\begin{equation}
    \braket{F(Z_1) - F(Z_2), Z_1-Z_2} \geq \eta \|Z_1-Z_2\|^2,
    \label{eq:mono1}
\end{equation}
for some constant $\eta>0$. (Here, the bracket is the dual pairing of $\mathcal{V}'\oplus\mathcal{V}'$ and $\mathcal{V}\oplus\mathcal{V}$. 

Furthermore, we need the concept of Lipschitz continuity: $F$ is locally Lipschitz with constant $L> 0$ if
\begin{align*}
    \Vert F(Z_1) - F(Z_2) \Vert \leq L \Vert Z_1-Z_2 \Vert .
\end{align*}
In particular, any (Fréchet) smooth function is locally Lipschitz continuous, and so are all its derivatives.

The map $F$ that we will study is the \emph{flipped gradient} of the general energy functional $\mathcal{E}:\mathcal{V}\oplus\mathcal{V}\to\mathbb{R}$, defined as
\begin{equation} 
    F(T,\Lambda) = (\partial_\Lambda \mathcal{E}(T,\Lambda), \partial_T \mathcal{E}(T,\Lambda)), 
    \label{eq:flipped-gradient}
\end{equation}
or more compactly $F(T,\Lambda) = R\partial \mathcal{E}(T,\Lambda)$, with $R$ being the map that exchanges the partial derivatives. The motivation is as follows: If we consider the bivariate Rayleigh quotient, $\partial \mathcal{E}_\text{bivar}$ is \emph{not} locally strongly monotone, as  its critical points are saddle points. On the other hand, the flipped gradient $F_\text{bivar} = R\partial\mathcal{E}_\text{bivar}$ can be seen to be locally strongly monotone near the \emph{ground state}, given that this ground state is non-degenerate with a finite spectral gap to the remaining spectrum. It is natural to expect that one can find conditions such that the flipped gradient of the energy when expressed in new coordinates is locally strongly monotone. 

The following is a central result, combining a result due to Zarantonello~\cite{Zarantonello1960,Zeidler1990} (points 1.~and 2.), adapted to the present notation and setting, and applied to the flipped gradient of an energy functional (point 3.).
\begin{theorem}\label{thm:z}
    Let $F :\mathcal{V}\oplus\mathcal{V}\to\mathcal{V}'\oplus\mathcal{V}'$ be a map, and let $U\subset{\mathcal{V}\oplus\mathcal{V}}$ be an open ball containing a $Z_*$ such that $F(Z_*) = 0$.

    Let $\mathcal{V}_d \subset \mathcal{V}$ be a Galerkin sequence of subspaces with $P_d$ being the orthogonal projector onto $\mathcal{V}_d\oplus\mathcal{V}_d$. Furthermore, let $F_d : \mathcal{V}_d\oplus\mathcal{V}_d \to \mathcal{V}_d'\oplus\mathcal{V}_d'$ be the Galerkin discretization of $F$, i.e., $F_d(Z_d) = P_d F(Z_d)$.
    
    Assume that $F$ is locally strongly monotone with constant $\eta>0$ and Lipschitz continuous with constant $L > 0$ on $U$.
    Then, the following holds:
    \begin{enumerate}
        \item \label{bullet1}
             $Z_*$ is the only root in $U$. 
    \item \label{bullet2}
          There is a sufficiently large $d_0$, such that for any $d>d_0$, there exists $Z_{d*}\in\mathcal{V}_d\oplus\mathcal{V}_d$ such that  $F_d(Z_{d*})=0$. This root is unique in $U$ and
          we have the following error estimate (quasi-optimality of the discrete solution): 
            \begin{equation} \|Z_{*d} - Z_*\| \leq \frac{L}{\eta} \mathrm{dist}(Z_*, \mathcal{V}_d\oplus\mathcal{V}_d). \label{eq:error-est1}
            \end{equation}

    \end{enumerate}
    Let $\mathcal E: \mathcal{V}\oplus\mathcal{V} \to \mathbb R$, $Z \mapsto \mathcal{E}(Z)$ be a (Fréchet) smooth energy functional. Let $R$ be the flipping map as introduced after Eq.~\eqref{eq:flipped-gradient} and set $F = R \partial \mathcal{E}$, and $E_* = \mathcal{E}(Z_*)$.
    \begin{enumerate}
        \setcounter{enumi}{2}
        \item \label{bullet3}
    For $d>d_0$, the discrete Galerkin equations $\partial \mathcal{E}_d(Z_{*d}) = 0$ have locally unique solutions, and in addition to the error estimate~\eqref{eq:error-est1}, we have the energy error
    \begin{align} 
        |\mathcal{E}(Z_{*d}) - E_*| &\leq C   \| Z_{*d} - Z_* \|^2 \nonumber \\
        &\leq C \left(\frac{L}{\eta} \right)^2 \mathrm{dist}(Z_*, \mathcal{V}_d\oplus\mathcal{V}_d)^2.
        \label{eq:error-est2}
    \end{align}
    \end{enumerate}
\end{theorem}

We will only present a partial proof of the theorem, as the proofs of points~\ref{bullet1}.~and \ref{bullet2}.~are standard, and can be found in, e.g., Ref.~\onlinecite{Zeidler1990}. Before we turn to the (short) proof of~\ref{bullet3}., we note that
Brouwer's fixed point theorem~\cite{Zeidler1986} can be used to obtain a sufficient condition for the constant $d_0$, where quadratic convergence sets in. In particular, if the Lipschitz and strong 
monotonicity constants are comparable, we here note that $d_0$ can be taken roughly such that $d>d_0$ implies
\begin{align}\label{eq:kappa}
\kappa(d)  =   \mathrm{dist} (Z_*, \mathcal V_d\oplus \mathcal V_d) < \frac{\delta \eta}{\eta + L} \approx \frac{\delta} 2,
\end{align}
where $\delta$ is the radius of the ball $U$. It should be noted, that this radius is unknown in general.

\begin{proof}[Proof of~\ref{bullet3}.]  For point \ref{bullet3}., we note that $F$ is locally Lipschitz continuous as a consequence of $\mathcal E$ being smooth, which together with strong monotonicity makes points~\ref{bullet1}. and~\ref{bullet2}.~applicable. The remaining argument follows~\cite{Laestadius2018} closely (where the case $\mathcal E = \mathcal E_\mathrm{NC-ECC}$ was treated). First, by assumption of $R$, $F(Z_*)=0$ and $F_d(Z_{*d})=0$ are equivalent to $\partial \mathcal E(Z_*)=0$ and $\partial \mathcal{E}_d(Z_{*d}) = 0$, respectively (note that $R$ commutes with $P_d$). Now,  Taylor expanding $\mathcal E$ around $Z_*$ and evaluating at $Z_{*d}$ gives
\begin{align*}
    \mathcal{E}(Z_{*d}) - E_* = \frac 1 2  \langle Z, \partial^2 \mathcal{E}(Z_*) Z \rangle  + \mathcal{O}(\|Z\|^3).
    %\partial^2 \mathcal{E}(Z_*)( (Z_{*d} - Z_*)^2) + \dots
\end{align*}
By the smoothness of $\mathcal E$, there exists a constant $C'$ such that
\[
 \langle Z, \partial^2 \mathcal{E}(Z_*) Z \rangle \leq C' \Vert Z\Vert^2.
\]
Further, the fact that on $U$ we can control the higher order terms by the quadratic one, we have
\[
\vert \mathcal{E}(Z_{*d}) - E_* \vert \leq  C \| Z_{*d} - Z_* \|^2.
\]
Using Eq.~\eqref{eq:error-est1} gives the full statement in Eq.~\eqref{eq:error-est2}
\end{proof}

The error estimate~\eqref{eq:error-est2} shows that for (smooth) energy functionals with a locally strongly monotone flipped gradient, the bivariational method of discretization behaves very similar to the usual Rayleigh--Ritz variational method of discretization. As we enlarge the Galerkin space, the discrete ground state converges, and the energy error is quadratic in the error of the state. However, we cannot guarantee convergence \emph{from above}, but this is much less important than actually having a quadratic error.

The following summarizes the main results of Ref.~\onlinecite{Laestadius2018}, where the proof and more details can be found:
\begin{theorem}[NC-ECC monotonicity]\label{thm:NC-ECCmono} 
    Assume that the system Hamiltonian $\hat{H}$ is self-adjoint, and that the ground-state of $\hat{H}$ exists, is non-degenerate, and that there is a spectral gap $\gamma>0$ between the ground-state energy $E_*$ and the rest of the spectrum. Assume that the reference $\phi_0$ is such that it is not orthogonal to the ground-state wavefunction. Let $Z_* = (T_*,\Lambda_*) \in \mathcal{V}\oplus\mathcal{V}$ be the corresponding critical point of $\mathcal{E}_{\text{NC-ECC}}$, and assume that $T_*$ and $\Lambda_*$ are not too large, i.e., that $\phi_0$ is a sufficiently good approximation to $\psi_*$. Then, $F = R \partial \mathcal{E}_\text{NC-ECC}$ is locally strongly monotone near $Z_*$ with a constant $\eta=C\gamma$, for some $C<1$.
\end{theorem}

\section{Monotonicity-preserving coordinate transformations}
\label{sec:coordinate-transformations}

\subsection{A class of exact coupled-cluster models}

In addition to the non-canonical ECC parameterization, Arponen also considered a second parameterization of the bra and ket wavefunctions, which gives equations of motion for the time-dependent Schrödinger equation that are canonical in the sense of Hamiltonian mechanics~\cite{Arponen1983,Arponen1987}. (This must not be confused with the use of canonical Hartree--Fock orbitals, which is unrelated.) This parameterization is given in terms of a \emph{coordinate transformation} $\theta_\text{C-ECC} : \mathcal{V}\oplus\mathcal{V} \to \mathcal{V}\oplus\mathcal{V}$ as
\begin{equation}
    (T,\Lambda) = \theta_{\text{C-ECC}}(T',\Lambda'), 
\end{equation}
where $\Lambda' = \Lambda$ and where $T = S(T';\Lambda')$ defined by
\begin{equation}
  Q T \phi_0 =
  Qe^{-{\Lambda'}^\dag} T'
  \phi_0, \quad Q = I - \ket{\phi_0}\bra{\phi_0}, \label{eq:canonical-coordinates}
\end{equation}
which has inverse $Q T'\phi_0 = Qe^{\tilde{\Lambda}^\dag} T\phi_0$. (In Arponen's work~\cite{Arponen1987}, the notation $(T',\Lambda')=(\Sigma,\tilde{\Sigma}^\dag)$ is used.)
The map $\theta_{\text{C-ECC}}$ is smooth and invertible with a smooth inverse, and we therefore obtain a new exact energy functional
\[ \mathcal{E}_{\text{C-ECC}} = \mathcal{E}_\text{NC-ECC}\circ \theta_{\text{C-ECC}}, \]
with values
\begin{equation}\label{eq:C-ECC} \mathcal{E}_{\text{C-ECC}}(T',\Lambda') = \braket{\phi_0, e^{(\Lambda')^\dag} e^{-S(T';\Lambda')} H e^{S(T';\Lambda')}\phi_0}.
\end{equation}
A remarkable consequence of this second parameterization is that it corresponds to retaining, in the perturbation series for the ground-state energy in terms of $(T,\Lambda)$, only those terms that can be represented by ``doubly linked'' diagrams~\cite{Arponen1983,Arponen1987},
\begin{equation}
    \mathcal{E}_{\text{C-ECC}}(T',\Lambda') = \braket{\phi_0, e^{(\Lambda')^\dag} e^{-T'}H e^{T'}\phi_0}_\text{DL}.
  \label{eq:ecc-doubly-linked}  
\end{equation}
This should be compared with Eq.~\eqref{eq:NC-ECC}. The phrase ``doubly linked'' means that every power of $(\Lambda')^\dag$ is connected to \emph{two} $T'$ operators, unless it is connected directly to $H$.
Thus, the canonical coordinates represent a more compact representation in that the resulting tensor contractions or diagrams in the energy are \emph{identical} to those obtained in the NC-ECC energy~\eqref{eq:NC-ECC}, except for some diagrams that are explicitly eliminated. 

Similarly, for the standard CC method, Arponen introduced the
coordinate transformation $\theta_{\text{CC}}$ given by
\begin{equation}
    (T,\Lambda) = \theta_{\text{CC}}(T',\Lambda') = (T', e^{\Lambda'}-1).
\end{equation}
We obtain the energy functional $\mathcal{E}_{\text{CC}} = \mathcal{E}_\text{NC-ECC} \circ \theta_{\text{CC}}$, where
\begin{equation}
    \mathcal{E}_{\text{CC}}(T',\Lambda') =  \braket{\phi_0, (1 + (\Lambda')^\dag) e^{-T'}He^{T'}\phi_0}, \label{eq:c-ecc}
\end{equation}
that is, the standard CC Lagrangian~\cite{Helgaker1988}. Incidentally, the standard CC coordinates are also canonical.

The map $\theta_{\text{CC}}$ can be generalized to Taylor polynomials. By setting
\[ e^{\Lambda} = (e^{\Lambda'})_n \equiv 1 + \Lambda' + \frac{1}{2}(\Lambda')^2 + \cdots + \frac{1}{n!} (\Lambda')^n, \]
we can solve for $\Lambda$ in terms of $\Lambda'$ by, e.g., considering first the singles, then doubles, etc., giving a smooth map $G_n : \mathcal{V} \to \mathcal{V}$ such that $e^{G_n(\Lambda')} = (e^{\Lambda'})_n$. In fact, since the cluster operators are nilpotent, $G_n(\Lambda') = \ln[(e^{\Lambda'})_n]$, where the logarithm is expanded in a (finite) Taylor series around the identity. Similarly, we can solve for $\Lambda$ in terms of $\Lambda'$, demonstrating that this map has an inverse, and in fact that this inverse is smooth. 
We obtain a coordinate transformation $\theta_{n}$ given by
\begin{equation}
\label{eq:theta-n}
(T,\Lambda) = \theta_n(T',\Lambda') = (T',G_n(\Lambda')), 
\end{equation}
and the corresponding energy functional
\begin{subequations}\label{eq:hierarchy}
\begin{equation}
    \label{eq:nc-ecc(n)}
    \mathcal{E}_{\text{NC-ECC}(n)}(T',\Lambda') = \braket{\phi_0, (e^{(\Lambda')^\dag})_n e^{-T'} H e^{T'} \phi_0}. 
\end{equation} 

Coordinate transformations form a group, and may thus be composed. By combining $\theta_{\text{C-ECC}(n)} = \theta_{n} \circ \theta_{\text{C-ECC}}$, we obtain an energy functional
\begin{equation}
    \label{eq:c-ecc(n)}
    \mathcal{E}_{\text{C-ECC}(n)}(T',\Lambda') = \braket{\phi_0, (e^{(\Lambda')^\dag})_n e^{-S(T';\Lambda')} H e^{S(T';\Lambda')}\phi_0}. 
\end{equation}
\end{subequations}
Since, in the NC-ECC energy functional~\eqref{eq:NC-ECC}, an exponential $e^{-\Lambda^\dag}$ can be inserted after $e^T$ without changing the result, both of these hierarchies correspond to truncations of a Baker--Campbell--Hausdorff expansion at order $n$, and are thus manifestly extensive.

\subsection{Coordinate transformation theorem}

Equations~\eqref{eq:nc-ecc(n)} and \eqref{eq:c-ecc(n)} represent two hierarchies of \emph{exact} parameterizations of the bivariate Rayleigh quotient. It is therefore of interest to determine whether they have locally strongly monotone flipped gradients. To establish this, we study the effect on local strong monotonicty of a coordinate transformation. 
\begin{theorem} \label{thm:map}
    Let $\mathcal{E} : \mathcal{V}\oplus\mathcal{V} \to \mathbb{R}$ be a smooth energy functional, let $Z_*$ be a critical point, and assume that $F = R\partial \mathcal{E}$ is locally strongly monotone near $Z_*$ with constant $\eta>0$.
    Let a smooth $\theta : \mathcal{V}\oplus\mathcal{V}\to\mathcal{V}\oplus\mathcal{V}$ with a smooth inverse be a given coordinate transformation, and let $\mathcal{E}_\theta = \mathcal{E} \circ \theta$ be the energy functional expressed in the new coordinates. Let $W_* = \theta^{-1}(Z_*)$ be the corresponding critical point for $\mathcal{E}_\theta$, and let $F_\theta=R \partial \mathcal{E}_\theta$ be its flipped gradient. Let $M_* = \partial \theta(W_*)$ be the Jacobian at $W_*$. Then we have the following conclusions:
    \begin{enumerate}
        \item If $M_*R = RM_*$, then $F_\theta$ is locally strongly monotone near $W_*$ with constant $\|M_*^{-1}\|^{-2}\eta$.
        \item In the noncommuting case, if $m_* = M_* - I$ is sufficiently small, $F_\theta$ is locally strongly monotone near $W_*$ with constant
        \[ \eta' = \eta \|M_*^{-1}\|^{-2} - C(I + \|m_*\|)\|m_*\|,\]
        where $C$ is the constant from Theorem~\ref{thm:z}(\ref{bullet3}).
    \end{enumerate}
\end{theorem}

\begin{proof}
With $F : \mathcal{V}\oplus\mathcal V\to  \mathcal V'\oplus \mathcal V'$ the flipped gradient
and $X \in \mathcal V \times \mathcal V$, we have
\begin{equation*}
  \braket{X, F(Z)}_{\mathcal V \oplus \mathcal V, \mathcal V'\oplus \mathcal V'} = \braket{RX, \partial \mathcal E(Z) }_{\mathcal V\oplus \mathcal V,\mathcal V'\oplus \mathcal V'}.
\end{equation*}
In the sequel, we omit the specification of the spaces in the dual pairing.

Let $Z_*  \in \mathcal{V}\oplus\mathcal{V}$ be such that $\partial \mathcal E (Z_*) = 0$, i.e., $F(Z_*)=0$. 
%By Lemma~1, it follows that $F_0$ is locally strongly monotone at $(t_*,\lambda_*)$. 
Since $F$ is smooth, local strong monotonicity of $F$ is equivalent to $\partial F(Z_*) \in
\mathcal{B}(\mathcal V\oplus \mathcal V,\mathcal V'\oplus\mathcal V')$ (a bounded linear operator) being coercive, i.e., there exists an $\eta_* >
0$ such that
\begin{equation*}
  \Delta(X) := \braket{X, \partial F(Z_*)X} \geq \eta_* \|X\|_{\mathcal{V}\oplus\mathcal{V}}^2.
\end{equation*}
(The constant $\eta$ in Eq.~\eqref{eq:mono1} approaches $\eta_*$ as the ball $U$ in the definition of local strong monotonicity approaches a point.)
To see this, we find an expression for $\Delta(h)$ in terms of the energy map,
\begin{align}
  \braket{X, F(Z_* + \epsilon X)}  
  = \braket{X, F(Z_*)  + \epsilon  F'(Z_*;X)}+ \mathcal O(\epsilon^2). \label{eq:eqv}
\end{align}
Here $F'(Z_*;X)$ is the directional derivative in the direction of $X$ such that
\begin{align*}
  \braket{X, \partial F(Z_*) X} &= \braket{X, F'(Z_*;X)}\\ 
  & = \frac{d}{d\epsilon}\braket{RX, \partial \mathcal E(Z_* + \epsilon X)}|_{\epsilon=0} \\
  &= \braket{RX, \partial^2 \mathcal E(Z_*) X}, 
\end{align*}
where $\partial^2\mathcal E(Z_*) \in \mathcal B(\mathcal V\oplus \mathcal V,\mathcal V'\oplus \mathcal V')$.
By choosing $\epsilon$ small enough, Eq.~\eqref{eq:eqv} and strong monotonicity gives the coercivity claim. The logical implication also goes in the reverse direction. (This will be used below.)

Recall that $\mathcal E_\theta =\mathcal E \circ \theta$ and that $F_\theta$ denotes the
flipped gradient of $\mathcal E_\theta$. We use that $F_\theta$ is locally strongly monotone at $W_* =
\theta^{-1}(Z_*)$ if and only if $\Delta_
\theta$ is coercive, i.e.,
\begin{equation*}
  \Delta_\theta(X) = \braket{X, \partial F_\theta(W_*) X} \geq \eta_\theta \|X\|^2,
\end{equation*}
for some $\eta_\theta>0$ and all $X \in \mathcal V\oplus \mathcal V$.
A straightforward application of the chain rule now gives
\begin{align*}
  \Delta_\theta(X) &= \braket{M_* RX, \partial^2 \mathcal E(Z_*)(M_* X)}, \\
  M_* &= \partial\theta(W_*) \in  \mathcal{B}(\mathcal V\oplus\mathcal V, \mathcal V\oplus \mathcal V).
\end{align*}
We note that this is \emph{almost} $\Delta(M_*X)$. Indeed,
\begin{equation}
  \begin{split}
  \Delta_\theta&(X) = \braket{RM_*X, \partial^2 \mathcal E(Z_*) (M_*X)}  \\
  &\quad + \braket{[M_*,R]X, \partial^2 \mathcal E(Z_*)(M_*X)} \\
&= \Delta(M_*X) +
  \braket{(M_*R-RM_*)X, \partial^2 \mathcal E(Z_*)(M_*X)}.\label{eq:MR}
\end{split}
\end{equation}
In particular, if $M_*R = R_*M$ then the last term vanishes in the utmost right-hand side of Eq.~\eqref{eq:MR}, and we obtain
monotonicity of $F_\theta$ but with a  modified constant. 

In the
case where $M_*R \neq RM_*$, we write $M_* = I + m_*$, and note that $M_*R - RM_*
= m_*R-Rm_*$. 
%(We could replace $I$ by any operator $M_0$ that commutes
%with $R$.) 
We obtain,
\begin{equation}\label{eq:Mm}
  \begin{split}
  \Delta_\theta(X) &\geq \eta \|M_*X\|^2 - \|\partial^2
  \mathcal E(Z_*)\|\|M_*\|\|m_*\|\|X\|^2 \\
 & \geq \Big[\eta \|M_*^{-1}\|^{-2} \\
&\quad - C\|(1 + \|m_*\|)\|m_*\| \Big]
 \|X\|^2.
\end{split}
\end{equation}
Here, we used that $\theta$ has a smooth inverse, implying $\|M_*X \| \geq
\|M_*^{-1}\|^{-1}\|X\|$, and that $\|M_*\| \leq I + \|m_*\|$. 
\end{proof}

\subsection{Monotonicity of (N)C-ECC$(n)$ models}

We apply Theorem~\ref{thm:map} to the maps $\theta_{n}$ and $\theta_n\circ \theta_{\text{C-ECC}}$ that define the NC-ECC$(n)$ and C-ECC$(n)$ models, respectively. The conclusion is as follows:
\begin{corollary}\label{cor:methods}
    For any of the NC-ECC$(n)$ or C-ECC$(n)$ models, the assumption that the ground-state critical point $W_* = (T'_*,\Lambda_*')$ is not too large is sufficient to guarantee local strong monotonicity of the flipped gradient of the energy, and hence a quasi-optimal solution to the Galerkin problem and a quadratic error estimate for the energy.
\end{corollary}
\begin{proof}
We consider the Jacobian of the coordinate map, which on block form reads
\begin{equation}
    \partial\theta(T',\Lambda') = 
    \begin{pmatrix}  \frac{\partial T}{\partial T'} & \frac{\partial T}{\partial\Lambda'} \\ \frac{\partial \Lambda}{\partial T'} & \frac{\partial \Lambda}{\partial \Lambda'} \end{pmatrix}.
\end{equation}
For the map $\theta_n$ (see Eq.~\eqref{eq:theta-n}), we first observe that
by definition,
\[ e^{\Lambda} = e^{\Lambda'} + \mathcal{O}(\|\Lambda'\|^{n+1}), \]
from which it follows, by taking the logarithm and expanding the logarithm around $e^{\Lambda'}$, which is a finite Taylor series,
\[ \Lambda = \Lambda' + \mathcal{O}(\|\Lambda'\|^{n+1}). \]
We obtain
\begin{equation}
    \partial\theta_n(T',\Lambda') = 
    \begin{pmatrix}  I & 0 \\ 0 & I + \mathcal{O}(\|\Lambda'\|^{n+1}) \end{pmatrix}.
\end{equation}
For the map $\theta_n\circ\theta_{\text{C-ECC}}$ we have, using the chain rule,
\begin{equation}
    \partial\theta_{\text{C-ECC}}(T',\Lambda') = \begin{pmatrix}
            A(\Lambda') & \partial_{\Lambda'} A(\Lambda')T \\ 0 & I+ \mathcal{O}(\|\Lambda'\|^{n+1}) 
    \end{pmatrix} .
\end{equation}
Here, $A(\Lambda')$ is the linear transformation  on $\mathcal{V}$ such that $S(T';\Lambda') = A(\Lambda')T'$ (see Eq.~\eqref{eq:canonical-coordinates}), i.e., $A(\Lambda')$ can be expressed in terms of the matrix representation of $e^{-(\Lambda')^\dag}$,
\begin{equation}
A(\Lambda') T' = \sum_{\mu,\nu\in \mathcal{I}} X_\mu \braket{\phi_\mu, e^{-(\Lambda')^\dag}\phi_\nu}\braket{\phi_\nu, T'\phi_0} .\label{eq:canon-trans-matrix-rep}
\end{equation}
We have $A(\Lambda') = I + \mathcal{O}(\|\Lambda'\|)$, and $\partial_{\Lambda'} A(\Lambda')T' = 
\mathcal{O}(\|T'\|)$. For both maps, the Jacobian of the coordinate transformation at the critical point $W_*=(T_*',\Lambda_*')$ becomes $M_* = I + m_*$ with $m_* = \mathcal{O}(\|W_*\|)$. Applying Theorem~\ref{thm:map}(2), the local strong monotonicity follows, and by Theorem~\ref{thm:z}, quasi-optimality of the truncated solutions and a quadratic error estimate.
\end{proof}

We note that our estimates are probably \emph{pessimistic} for some of the models covered here. The analysis starts with a given monotonicity constant $\eta$ for the NC-ECC scheme, and consistently produces an $\eta'<\eta$ for the method obtained using the coordinate change, worsening the error estimates. However, it may well be that a direct analysis of the secondary method yields a better $\eta'$. However, the important point here is that Theorem~\ref{thm:map} does guarantee that the new method is convergent under \emph{some} reasonable conditions. For example, we have now proven that quadratic coupled-cluster (QCC) theory~\cite{VanVorhiis2000a} is convergent if the reference is a sufficiently good approximation to the ground state, and using Eq.~\eqref{eq:kappa} also a basic means to study which truncations or Galerkin schemes can be reasonable, at least in principle. It may be interesting to see whether truncation schemes like the PP hierarchy with orbital optimization, also in a quadratic $n=2$ or higher formulation~\cite{Byrd2002}, can be further analyzed based on our results.

In the proof of Theorem~\ref{thm:map}, it arises naturally that the most favorable coordinate transformations are those that commute with the flipping map $R$, since local strong monotonicity then follows with no assumptions on the Jacobian of the map. The Jacobian commutes with $R$ if and only if
\begin{equation}
    \frac{\partial T}{\partial T'} = \frac{\partial \Lambda}{\partial \Lambda'}, \quad \text{and} \quad \frac{\partial T}{\partial \Lambda'} = \frac{\partial \Lambda}{\partial T'},
\end{equation}
and one must assume that this holds at every point in $\mathcal{V}\oplus \mathcal{V}$, as one does not know \emph{a priori} where the critical point is. It is not clear what such transformations in general look like, and whether such transformations are useful reparameterizations of the energy.

\subsection{Properties of the canonical and non-canonical schemes}

The coordinate transformation $\theta_{\text{C-ECC}}$ as represented by Eq.~\eqref{eq:canon-trans-matrix-rep} is such that when applied to a cluster operator $T'_k$ of rank $k$, it generates terms $T_{k'}$ with $k'\leq k$. The same is true for the inverse map. Thus, if $\mathcal{V}_d$ is excitation-rank complete, i.e., it contains all excitations of rank up to and including some $k > 0$, then the Galerkin discretization of the NC-ECC($n$) method commutes with changing coordinates via the coordinate map $\theta_{\text{C-ECC}}$,
\[ \mathcal{E}_{\text{C-ECC($n$)},d} = \mathcal{E}_{\text{NC-ECC($n$)},d} \circ \theta_{\text{C-ECC}}. \]
By inpection, one can also see that this holds for a doubles-only truncation, since $A(\Lambda')=I$ in this case.
We obtain the following result:
\begin{theorem}\label{thm:commute}
    Let $\mathcal{V}_d$ be excitation-rank complete or consist of doubles excitations only.  Then, the discrete solutions $(T_{*d},\Lambda_{*d})$ and $(T_{*d}',\Lambda_{*d}')$ of the NC-ECC($n$) and C-ECC($n$) methods, respectively, are equivalent and related via $\theta_{\text{C-ECC}}$, i.e.,
$T_{*d} = A(\Lambda_{*d}') T_{*d}'$ and $\Lambda_{*d} = \Lambda_{*d}'$. 
\end{theorem}
We stress that, if $\mathcal{V}_d$ is \emph{not} excitation-rank complete, the canonical and non-canonical parameterizations are not equivalent. This would be the case for the PP hierarchy of truncations~\cite{Byrd2002,Lehtola2016}.

According to the doubly linked structure of the energy functional $\mathcal{E}_{\text{C-ECC($n$)}}$, see the discussion after Eq.~\eqref{eq:ecc-doubly-linked}, the amplitude equations for the canonical case are cheaper, albeit by a small amount. Moreover, it is reasonable to expect that the canonical solution $(T'_{*d},\Lambda'_{*d})$ is more compact compared to $(T_{*d},\Lambda_{*d})$. We investigate this claim numerically in Section~\ref{sec:numerics}.

\section{Numerical Results}
\label{sec:numerics}

\subsection{Implementation}
The (un-)truncated (N)C-ECC($n$)SD equations~\eqref{eq:hierarchy} and~\eqref{eq:ampeq-h}, together with the untruncated coordinate transformation Eq.~\eqref{eq:canon-trans-matrix-rep} have been implemented in a local full CI-based program, i.e., all intermediates are expressed as vectors in the full CI basis. To this end, the C-ECC($n$) amplitudes are computed using transformed NC-ECC($n$) residual expressions~\cite{Evangelista2011}. The C-ECC($n$) amplitude equations are thus:
\begin{subequations}\label{eq:num-amp-eq}
\begin{align}
  \begin{split}
  0&=\sum_{\nu\in\mathcal{I}} \langle \phi_\nu ,  e^{-(\Lambda')^\dagger} \phi_\mu\rangle  \langle \phi_0 , (e^{(\Lambda')^\dagger})_n  [H^S,X_\nu]  \phi_0\rangle, \label{CECCamp1} 
  \end{split}\\
    \begin{split}
0&=  \langle \phi_\mu ,  (e^{(\Lambda')^\dagger})_{n-1}  H^S \phi_0\rangle \\&- \sum_{\nu\in\mathcal{I}}  \langle \phi_\mu , X_\nu^\dagger  e^{-(\Lambda')^\dagger} T' \phi_0\rangle  \langle \phi_0 , (e^{(\Lambda')^\dagger})_n [H^S,X_\nu]  \phi_0\rangle ,
    \label{CECCamp2}
    \end{split}
    \end{align}
    \end{subequations} 
where $H^S = e^{-S(T';\Lambda')} H e^{S(T';\Lambda')}$.  The sparsity of the coordinate transformation Eq.~\eqref{eq:canon-trans-matrix-rep} has been exploited throughout, e.g., only singles amplitudes are transformed in the case where $\mathcal{V}_d$ contains only singles and doubles.
 The coupled amplitude equations~\eqref{eq:num-amp-eq}
 are solved iteratively starting from an MP2-guess, using an alternating scheme and applying DIIS convergence acceleration. In all computations,  residuals and energies were converged to a threshold of $10^{-4}$ and  $10^{-6}$ a.u., respectively. The  (N)C-ECC($1$)SD and (N)C-ECC($\infty$)SD implementations are verified by reproducing the ``CCSD'' and ``ECCSD'' energies presented in~\cite{Evangelista2011}.

\subsection{Numerical Experiments}

The (N)C-ECC($n$)SD and (N)C-ECC($n$)DT models have been studied numerically by investigating the potential energy curves of the hydrogen fluoride molecule with intermolecular distances $1.0 \le R \le 3.5~(a_0)$ in a DZV basis set~\cite{Evangelista2011} as well as the H$_8$ model system with structural parameters $0.0001 \le \alpha \le 1.0$ in a MBS basis set~\cite{Jankowski1985,Adamowicz2000}. 
For large distances $R$ and small $\alpha$, respectively, the systems comprise multireference character, i.e., the weight    of the Hartree--Fock configuration in the full CI wave function,  $|\braket{\phi_0,\psi^\text{FCI}}|^2$,  is fairly small. Thus, these species are good candidates to study novel quantum chemical methods.

The energy curves of the canonical models C-ECC($n$)SD are identical to the non-canonical NC-ECC($n$)SD ones and are thus not presented here. However, the results differ if excitation-rank incomplete truncation schemes are employed, e.g., when orbital optimization is considered. For instance, in canonical ECC($n$)DT, singles amplitudes are effectively generated from doubles and triples amplitudes, while these are absent in the non-canonical model. This effect has been studied on the potential curve of the H--F molecule and is depicted in Fig.~\ref{fig:hf-num-dt}:
The generation of singles amplitudes entails that the canonical computation is lower in energy, in particular towards the multireference region where these contribute significantly to the wave function expansion. This depends, however, on the role of the singles amplitudes in the wave function: In test computations on the H$_8$ model system 
a different trend was observed, consistent with the diminished importance of singles in the wave function (\emph{vide infra}). Therefore, we cannot conclude that the canonical coordinates are consistently better when unconventional truncations are used.

\begin{figure}[htb]
\centering

\includegraphics{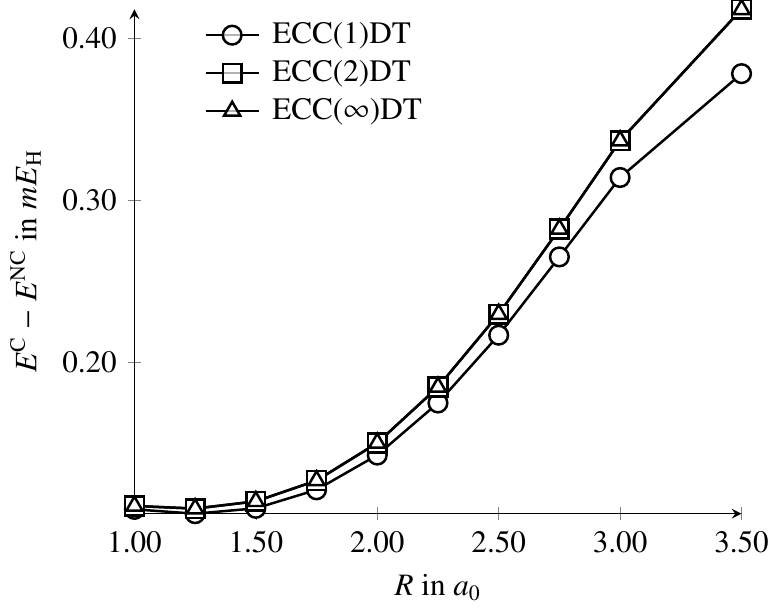}
\caption{Difference between a canonical and a non-canonical ECC($n$)DT computation for the potential curve of H--F.
}
\label{fig:hf-num-dt}
\end{figure}

In order to investigate the effect of using different coordinates in  ECC($n$)SD computations,
we calculated a set of CC diagnostics which are often used to assess the quality of CC computations~\cite{Lee1989}. These are based on the largest singular value ($\mathcal{D}_1$) and Frobenius norm ($\mathcal{T}_1$) of the matrix representation of the singles amplitudes. (Equivalently, $\mathcal{T}_1^2 = \|\tau_{*1}\|_2^2/N$, the sum of the squares of the singles amplitudes, with $N$ the number of correlated electrons.) Though diagnostics based on  doubles amplitudes are preferred, they are not as  available in implementations as are the singles-based variants~\cite{Jiang2012}. Additionally, we computed the diagnostic $||\tau_\ast ||_2^2 /(||\tau_\ast ||_2^2+1)$ which involves all the amplitudes. This choice can be motivated from monotonicity arguments and will be discussed in a forthcoming paper. 

The values have been computed for truncation schemes $n=1,2,\infty$. Since the values are very similar, only the data for $n=1$ is presented. Fig.~\ref{fig:diagn-compare-1} shows the diagnostics correlated with the multireference character for the H--F potential curve, in Fig.~\ref{fig:h8diagn-compare-1} values for the H$_8$ model are shown. In H$_8$, electron correlation is dominated by doubles amplitudes, as can be seen from the small values of the singles based diagnostics.  Since the reparameterization of the amplitudes in the canonical model does not affect the amplitudes of highest excitation rank, the difference between the NC-ECC(1)SD and C-ECC(1)SD amplitude vectors is negligible. This is different for the H--F case. Here, the amplitude norms of the canonical models are consistently smaller than the non-canonical variants, indicating that the wave function parameterization is more compact.

\begin{figure}[htb]
\centering

\includegraphics{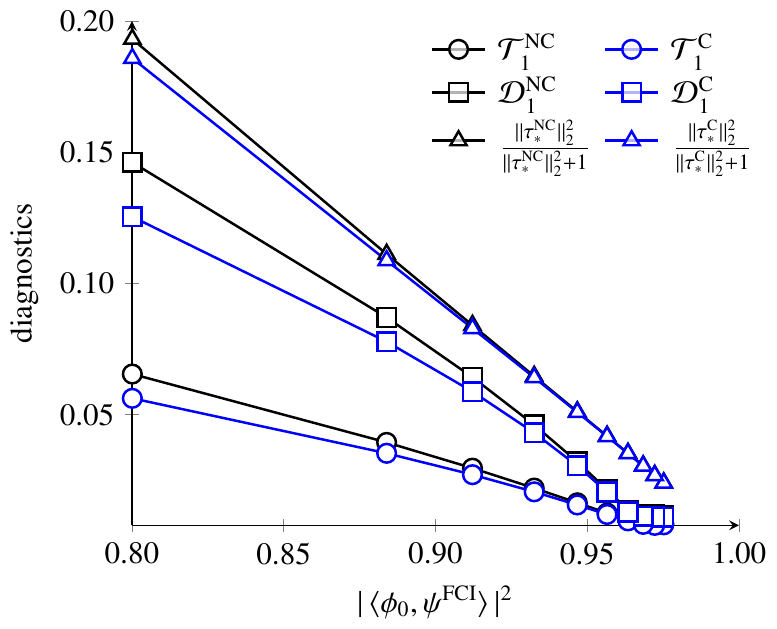}

\caption{Comparison of CC diagnostics of the C-ECC($1$)SD and NC-ECC($1$)SD model for the H--F potential curve correlated with the multireference character.}
\label{fig:diagn-compare-1}
\end{figure}

\begin{figure}[htb]
\centering

\includegraphics{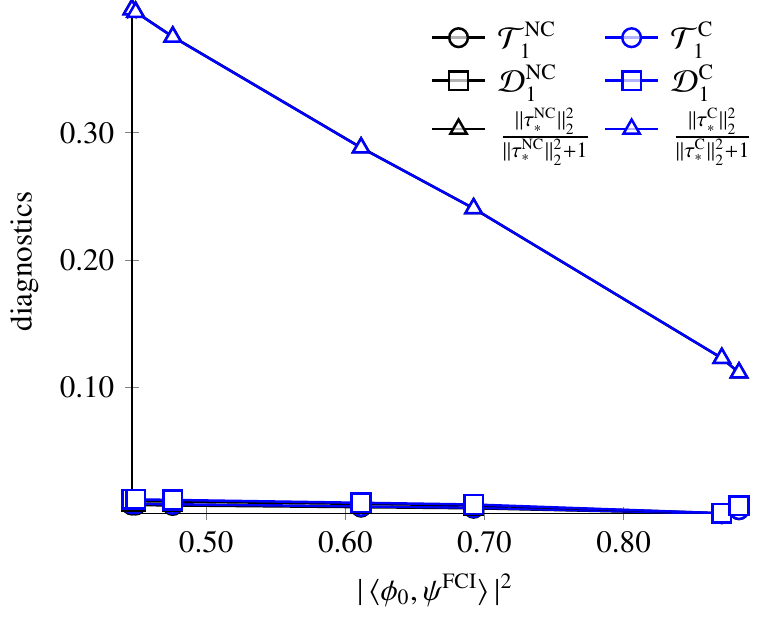}

\caption{Comparison of CC diagnostics of the C-ECC($1$)SD and NC-ECC($1$)SD model for the H$_8$ potential curve correlated with the multireference character.}
\label{fig:h8diagn-compare-1}
\end{figure}

Our numerical experiments suggest, that for excitation-rank incomplete models, the canonical map generates effectively an excitation-rank complete parameterization, but does not necessarily yield significantly better results. Concerning the \emph{a priori} excitation-rank complete models, it has been found that the canonical parameterization can be more compact compared to the non-canonical one, a desired property for post-Hartree--Fock methods.

\section{Concluding remarks}
\label{sec:conclusion}

In this article, we formulated basic error estimates for a class of exact models, defined in terms of replacing, in Arponen's ECC method, the exact exponential $e^{\Lambda^\dag}$ of the dual cluster operator with a finite-order Taylor polynomial, the canonical C-NCC($n$) models and the non-canonical NC-ECC($n$) models. The central result was a coordinate-transformation theorem, Theorem~\ref{thm:map}, that gives error estimates for any method that can be described as a coordinate transformation of ECC theory. Notably, these results guarantee asymptotically quadratic error estimates for the ground-state energy of all models, under certain mild conditions.

Apart from Theorem 3, a basically self-contained mathematical framework for local error analysis of coupled-cluster methods was presented. This was based on Arponen's bivariational principle and basic results from nonlinear monotone operator theory, i.e., Zarantonello's theorem. Also central was our prior analysis of Arponen's extended coupled-cluster method in its noncanonical formulation. 

The methods covered by our analysis include standard CC theory, quadratic CC theory~\cite{VanVorhiis2000a,Byrd2002}, the perfect-pairing hierarchy~\cite{Lehtola2016} for approximating CASSCF, also in its quadratic version~\cite{Byrd2002}, and, Arponen's canonical ECC method.

The error estimates are not optimal for many methods. A direct analysis of canonical ECC would probably provide the most optimistic analysis for all the (N)C-ECC($n$) methods, due to the doubly linked structure and the equivalence of excitation-rank complete Galerkin discretizations.

Finally, we performed some simple numerical experiments, focusing on the possibility of using canonical coordinates in place of the usual CC amplitudes when doing diagnostic estimates on CC calculations on systems with multireference character. Our preliminary findings support the hypothesis that the canonical coordinates are more compact compared to the usual coordinates, providing more accurate diagnostics.

An interesting extension of the present work would be to study truncations where singles-amplitudes are replaced by orbital rotations, either unitary or biorthogonal, as in the QCC and PP approaches, or the non-orthogonal orbital optimized coupled-cluster method of Pedersen and coworkers.~\cite{Pedersen2001}. Moreover, the complete-active space coupled-cluster method by Adamowicz and coworkers~\cite{Adamowicz2000} fits the present scheme. It is also known that quadratic CC and ECC in general are quite good at reproducing multireference character, while standard single-reference CC is quite poor at this. Thus, a modified analysis of the ECC method  that includes multireference assumptions, such as the steerable CAS-ext gap of Ref.~\onlinecite{faulstich2018}, could potentially lead to a deeper understanding of how CC methods generally behave in the presence of static correlation.

\acknowledgments

This work has received funding from the Research Council of Norway (RCN) under
CoE Grant Nos. 287906 and 262695 (Hylleraas Centre for Quantum Molecular Sciences), and from ERC-STG-2014
under grant No. 639508. The authors are thankful to M.~A.~Csirik for useful comments.

\bibliographystyle{unsrt}
\bibliography{references.bib}

\end{document}